\documentclass[letterpaper, 10 pt, conference]{ieeeconf}  

\IEEEoverridecommandlockouts                              
\usepackage{amsmath} 
\usepackage{amssymb}  
\usepackage{mathrsfs}
\usepackage{mathtools}
\usepackage{bm}
\usepackage{graphicx}
\usepackage{cite}
\usepackage{subfigure}

\newcommand{\figref}[1]{Fig.~\ref{#1}}

\newcommand{\lemref}[1]{Lemma~\ref{#1}}
\newcommand{\thmref}[1]{Theorem~\ref{#1}}
\newcommand{\defref}[1]{Definition~\ref{#1}}
\newcommand{\secref}[1]{Section~\ref{#1}}

\newcommand{\R}{\mathbb{R}}
\renewcommand{\S}{\mathcal{S}}
\newcommand{\K}{\mathcal{K}}
\newcommand{\M}{\mathcal{M}}
\newcommand{\C}{\mathcal{C}}
\newcommand{\KeyGen}{\mathsf{KeyGen}}
\newcommand{\Enc}{\mathsf{Enc}}
\newcommand{\Dec}{\mathsf{Dec}}
\newcommand{\pk}{\mathsf{pk}}
\newcommand{\sk}{\mathsf{sk}}
\newcommand{\A}{\mathcal{A}}
\newcommand{\B}{\mathcal{B}}
\renewcommand{\Game}{\mathsf{Game}}
\newcommand{\Adv}{\mathsf{Adv}}
\newcommand{\INDCPA}{\mathsf{IND-CPA}}
\newcommand{\INDPEA}{\mathsf{IND-PEA}}
\newcommand{\sample}{\overset{\$}{\gets}}

\newtheorem{definition}{Definition}

\newtheorem{lemma}{Lemma}
\newtheorem{theorem}{Theorem}

\newtheorem{remark}{Remark}

\title{\LARGE \bf
    Towards Provably Secure Encrypted Control Using \\ Homomorphic Encryption$^\ast$
}

\author{Kaoru Teranishi$^{1,2}$ and Kiminao Kogiso$^{1}$
\thanks{$^{\ast}$This work was supported by JSPS Grant-in-Aid for JSPS Fellows Grant Number JP21J22442 and for Scientific Research (B) Grant Number JP22H01509.}
\thanks{$^{1}$Department of Mechanical and Intelligent Systems Engineering,
        The University of Electro-Communications, 1-5-1 Chofugaoka, Chofu, Tokyo 1828585, Japan
        {\tt\small \{teranishi, kogiso\}@uec.ac.jp}}%
\thanks{$^{2}$Research Fellow of Japan Society for the Promotion of Science}%
}

\begin{document}

\thispagestyle{empty}
\hspace{-4.5mm}
\fbox{
\begin{minipage}{\textwidth-5mm}\scriptsize
© 20XX IEEE.  Personal use of this material is permitted.  Permission from IEEE must be obtained for all other uses, in any current or future media, including reprinting/republishing this material for advertising or promotional purposes, creating new collective works, for resale or redistribution to servers or lists, or reuse of any copyrighted component of this work in other works.
\end{minipage}
}
\newpage
\setcounter{page}{0}

\maketitle
\thispagestyle{empty}
\pagestyle{empty}

\begin{abstract}

Encrypted control is a promising method for the secure outsourcing of controller computation to a public cloud.
However, a feasible method for security proofs of control has not yet been developed in the field of encrypted control systems.
Additionally, cryptography does not consider certain types of attacks on encrypted control systems; therefore, the security of such a system cannot be guaranteed using a secure cryptosystem.
This study proposes a novel security definition for encrypted control under attack for control systems using cryptography.
It applies the concept of provable security, which is the security of cryptosystems based on mathematical proofs, to encrypted control systems.
Furthermore, this study analyzes the relation between the proposed security and the conventional security of cryptosystems.
The results of the analysis demonstrated that the security of an encrypted control system can be enhanced by employing secure homomorphic encryption.

\end{abstract}

\section{Introduction}\label{Introduction}

The encrypted control method is based on the homomorphism of cryptosystems that enables implementation of controller computation in an encrypted manner~\cite{Kogiso15,Farokhi17,Kim16,Darup21}.
The controller parameters and inputs/outputs of an encrypted controller are encrypted against attackers originating outside the closed-loop system.
This control method can be implemented for the secure outsourcing computation of controllers to a public cloud because the control inputs are computed without using a decryption key.

Although encrypted control is a promising method for secure control, research on the security of encrypted control systems remains limited.
In most of the previous studies, an encrypted control system was assumed to be secure if the underlying cryptosystem was secure.
However, certain attacks against encrypted control systems are not considered in cryptography.
\figref{fig:net_enc} depicts a typical situation in standard signal encryption.
Alice encrypts a message and transmits the encrypted message to Bob, who decrypts it.
Eve, an attacker, intercepts the encrypted data and tries to recover it.
Eve eavesdrops on the input and output of the encrypted controller while attacking the encrypted control system, as shown in \figref{fig:enc_ctrl}.
This is because control systems are typically closed-loop systems.
In this case, Eve's objective is not necessary to obtain the original controller input and output.
For example, Eve may attempt to recover an encrypted controller parameter from encrypted messages.
This objective may be achievable even if Eve fails to decrypt the encrypted input and output.
Although this type of attack on encrypted control systems is fairly common, it is not considered in standard cryptosystems, including the ElGamal encryption~\cite{ElGamal85} and Paillier encryption~\cite{Paillier99}, which are widely used for encrypted control.
Therefore, using a secure cryptosystem does not necessarily mean that an encrypted control system is secure.

\begin{figure}[t]
    \centering
    \subfigure[Eavesdropping in encrypted network.]{\includegraphics[scale=.85]{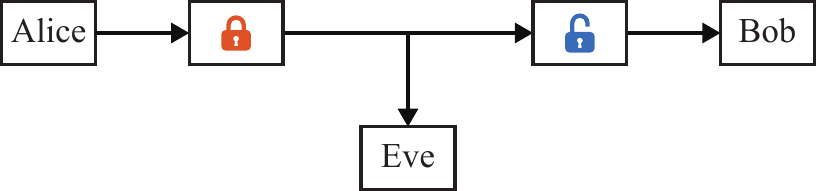}\label{fig:net_enc}}
    \subfigure[Eavesdropping in encrypted control.]{\includegraphics[scale=.85]{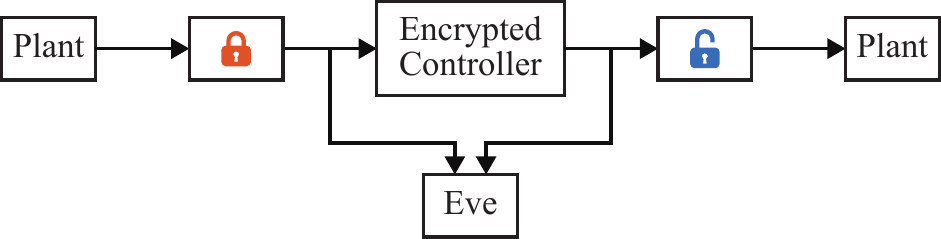}\label{fig:enc_ctrl}}
    \caption{Eavesdropping in encrypted network and encrypted control.}
    \label{fig:eavesdropping}
    \vspace{-5mm}
\end{figure}

Furthermore, the definition of security for encrypted control systems is ambiguous, and the types of attacks against the systems are not clearly defined.
Meanwhile, the security of cryptosystems has been mathematically defined and analyzed in the field of cryptography~\cite{Katz21}.
One of the security approaches, called \textit{provable security}, is formulated based on the objectives and capabilities of the attackers.
Additionally, it is proved under some assumptions of computational complexity.
More precisely, the security is usually demonstrated using proof by contradiction, namely if there exists an algorithm to break a cryptosystem, then a practically difficult computational problem can be efficiently solved by using the algorithm.
This proof construction is called the reductionist approach and is fundamental in proving the security of modern cryptosystems.
The security of encrypted control systems can be clarified and formally analyzed by applying a similar security definition and proof construction.
Consequently, the standard security definition of cryptosystems must be expanded to encompass the specific attacks for encrypted control systems because the attack scenarios considered in cryptography and control systems vary significantly.
Some studies have already applied cryptographic security to encrypted control protocols.
The privacy of cloud-based optimization, model predictive control, and state estimation is achieved through the computational indistinguishability of two random ensembles~\cite{Alexandru21_1,Alexandru20_3,Ristic21}, and that of distributed and cooperative controls is achieved based on a cryptographic game~\cite{Alexandru20_1,Alexandru19_3}.
However, the guaranteed privacy of the aforementioned studies is based on the input/output privacy of the systems rather than a controller parameter and is provided for specific protocols.

This study proposes a novel security definition of encrypted control systems under a parameter estimation attack via a game based on a cryptographic security notion.
In a parameter estimation attack, an attacker attempts to disclose the original controller parameter of an encrypted controller from multiple data of encrypted input and output.
The proposed security is defined as an attacker obtaining no information about a controller parameter by the attack.
Additionally, we analyze the relation of the proposed security to the standard security of cryptosystems.
Contrary to the existing studies~\cite{Alexandru21_1,Alexandru20_3,Ristic21,Alexandru20_1,Alexandru19_3}, this study aims to expand the application of security notion in cryptography to encrypted control systems.
Moreover, it contributes to the establishment of a methodology of security proofs for general encrypted control.
The proposed security definition allows for the analysis of the security of a broader encrypted control system under a parameter estimation attack.
The analysis of this study clarifies the strength of the proposed security.
It is shown that the proposed security is weaker than or equal to the standard security of cryptosystems.
This implies that a secure encrypted control system under the attack can be realized using a secure cryptosystem that satisfies the standard cryptographic security.
In addition, this study discusses a condition in which the proposed security is as secure as the standard cryptosystem security.
The condition suggests that the proposed security for most controls may be achievable by a simple and lightweight scheme rather than a cryptosystem that satisfies the standard security.

The remainder of this paper is organized as follows.
\secref{sec:preliminaries} introduces homomorphic encryption and encrypted control employed in this study.
In addition, the concept of provable security in cryptography is provided.
\secref{sec:provable_security} proposes a novel security definition for encrypted control and presents the analysis results.
\secref{sec:further_development} presents the scope for further development of the proposed security definition.
\secref{sec:conclusions} presents the conclusions and future work.

\section{Preliminaries}\label{sec:preliminaries}

\subsection{Notation}

The set of real numbers is denoted by $\R$.
The sets of $n$-dimensional column-vectors and $m$-by-$n$ matrices of which elements belong to a set $X$ are denoted by $X^{n}$ and $X^{m\times n}$, respectively.
An algorithm $\A$ is polynomial-time algorithm if there exists a polynomial $p$ such that, for every $k$-bit length input $x$, $\A(x)$ terminates within at most $p(k)$ steps~\cite{Katz21}.
We say a function $\epsilon:\{1,2,3,\dots\}\to\R$ is negligible if, for every positive integer $c>0$, there exists an integer $N$ such that $|\epsilon(n)|<n^{-c}$ holds for all $n>N$~\cite{Katz21}.

\subsection{Homomorphic encryption}

This section summarizes the basics of homomorphic encryption.
Homomorphic encryption is a cryptosystem that enables the application of arithmetic over a ciphertext space.
The definition of a cryptosystem is as follows.

\begin{definition}[Cryptosystem~\cite{Buchmann01}]\label{def:cryptosystem}
    A cryptosystem is defined as a tuple $\Pi=(\KeyGen,\Enc,\Dec)$, which satisfies the following properties:
    \begin{itemize}
        \item $\KeyGen:\S\to\K$ is a key-generation function that determines keys $(\pk,\sk)\in\K$ from a security parameter $\lambda\in\S$, where $\K$ is a key space, and $\S$ is a set of security parameters.
        \item $\Enc_{\pk}:\M\to\C$ is an encryption function that outputs a ciphertext $c\in\C$ from a plaintext $m\in\M$, where $\M$ is a plaintext space, and $\C$ is a ciphertext space.
        \item $\Dec_{\sk}:\C\to\M$ is a decryption function that outputs a plaintext $m\in\M$ from a ciphertext $c\in\C$.
        \item For any $\pk$, there exists $\sk$ such that $\Dec_{\sk}(\Enc_{\pk}(m))=m$ for all $m\in\M$.
    \end{itemize}
    Further, we often say $\Pi$ is symmetric-key encryption if $\pk=\sk$; otherwise it is public-key encryption.
\end{definition}

Homomorphic encryption is defined as follows.

\begin{definition}[Homomorphic encryption]\label{def:he}
    A cryptosystem $\Pi$ in \defref{def:cryptosystem} is called homomorphic encryption if there exist binary operations $\bullet:\M\times\M\to\M$ and $\circ:\C\times\C\to\C$ such that $\Dec_{\sk}(\Enc_{\pk}(m)\circ\Enc_{\pk}(m'))=m\bullet m'$ for all $m,m'\in\M$.
\end{definition}

Homomorphic encryption is classified based on the allowed arithmetic, into additive, multiplicative, somewhat, and fully homomorphic encryption.
The additive and multiplicative homomorphic encryption satisfy $\Dec_{\sk}(\Enc_{\pk}(m)\oplus\Enc_{\pk}(m'))=m+m'$ and $\Dec_{\sk}(\Enc_{\pk}(m)\otimes\Enc_{\pk}(m'))=mm'$, respectively.
Somewhat homomorphic encryption can perform the arithmetic $\oplus$ and $\otimes$ only a limited number of times.
Fully homomorphic encryption allows the computation of any number of arbitrary arithmetic.

\subsection{Encrypted control}

Various encrypted control frameworks have been proposed in conventional works according to the controller type, such as static output feedback~\cite{Farokhi17}, cooperative control~\cite{Darup19_1,Alexandru19_3}, and model predictive control~\cite{Darup18_2,Alexandru18}.
However, this study does not assume a specific type of controller to be encrypted.
A controller is considered as a map from a plant output to the control input to analyze the cryptographic security for as broad a class of encrypted controllers as possible.

\begin{definition}[Encrypted control]
    Given a control law
    \begin{equation}
        \bm{u}=f(\bm{y};\bm{K}),
        \label{eq:control}
    \end{equation}
    where $\bm{u}\in\M^{q}$ is a control input, $\bm{y}\in\M^{\ell}$ is a plant output, and $\bm{K}\in\M^{r}$ is a controller parameter.
    Let $\Pi$ be homomorphic encryption in \defref{def:he}.
    We say a map $f_{\Pi}$ is an encrypted control law of $f$ with $\Pi$ if \[
        \Dec_{\sk}(f_{\Pi}(\Enc_{\pk}(\bm{y});\Enc_{\pk}(\bm{K}))=f(\bm{y};\bm{K})
    \]
    holds, where for a plaintext vector $\bm{m}=[m_{1}\ \cdots\ m_{k}]^{\top}\in\M^{k}$ and ciphertext vector $\bm{c}=[c_{1}\ \cdots\ c_{n}]^{\top}\in\C^{n}$ the encryption and decryption functions perform each element of the vectors, namely $\Enc_{\pk}(\bm{m})=[\Enc_{\pk}(m_{1}) \ \cdots \ \Enc_{\pk}(m_{k})]^{\top}$ and $\Dec_{\sk}(\bm{c})=[\Dec_{\sk}(c_{1}) \ \cdots \ \Dec_{\sk}(c_{n})]^{\top}$.
\end{definition}

\begin{remark}
    In practice, $\bm{y}$ and $\bm{K}$ are given by real-valued vectors, and $\M$ is not necessarily the same as $\R$.
    Therefore, the elements of $\bm{y}$ and $\bm{K}$ must be encoded into plaintext before encryption~\cite{Farokhi17,Teranishi19_3,Darup20_2}.
    This encoding process is omitted in this study because it is focused on the cryptographic properties of the encrypted controls.
\end{remark}

\subsection{Provable Security}

The security of modern cryptosystems is demonstrated through mathematical proofs using the model of an attacker, and the security proved by such formal procedures is called \textit{provable security}.
A popular method to define the security is a \textit{game} between an attacker and a challenger.
The game simulates a situation in which an attacker attempts to break a cryptosystem using an \textit{oracle}, which is a black box that outputs an ideal response for an input.
The oracle represents the capabilities of the attacker.
It should be noted that specific attack methodologies for cryptosystems, such as brute-force and side-channel attacks, are not considered in provable security.
In the framework of provable security, attackers are considered as probabilistic polynomial-time algorithms.

This section introduces \textit{indistinguishability under chosen plaintext attack (IND-CPA)}, which is a traditional security concept for public-key encryption, as an example of provable security.

\begin{definition}[IND-CPA~\cite{Katz21}]\label{def:ind-cpa}
    Let $\Pi\!=\!(\KeyGen,\Enc,\Dec)$ be public-key encryption.
    Define a game between an attacker $\A=(\A_{1},\A_{2})$ and a challenger as follows.
    \begin{center}
        \fbox{
            \begin{minipage}{.9\columnwidth}
                $\Game^{\INDCPA}_{\Pi,\A}$:
                \begin{enumerate}
                    \renewcommand{\labelenumi}{\arabic{enumi}.}
                    \item $(\pk,\sk)\gets\KeyGen(\lambda)$
                    \item $(m_{0},m_{1},\sigma)\gets\A_{1}(\pk)$
                    \item $b\sample\{0,1\}$
                    \item $c\gets\Enc_{\pk}(m_{b})$
                    \item $\hat{b}\gets\A_{2}(c,\sigma)$
                \end{enumerate}
            \end{minipage}
        }
    \end{center}
    \begin{itemize}
        \item\textbf{Setup}:
        The challenger computes keys $\pk$ and $\sk$ from a security parameter $\lambda$ using the key-generation function $\KeyGen$ and gives $\pk$ to the attacker (Line~$1$).
        \item\textbf{Challenge}:
        The attacker computes two different plaintexts $m_{0},m_{1}\in\M$, which are of identical size, and an intermediate state $\sigma$ using the polynomial-time algorithm $\A_{1}$ (Line~$2$).
        The attacker gives $m_{0}$ and $m_{1}$ to the challenger.
        The challenger randomly selects a bit $b\in\{0,1\}$ and computes a ciphertext $c$ of $m_{b}$ using the encryption function $\Enc$ with $\pk$ (Lines~$3$--$4$).
        The challenger gives $c$ to the attacker.
        \item\textbf{Guess}:
        The attacker guesses the plaintext that has been encrypted from the given ciphertext $c$ using the polynomial-time algorithm $\A_{2}$ and outputs a bit $\hat{b}\in\{0,1\}$ (Line~$5$).
        If $\hat{b}=b$, the attacker wins the game; otherwise, the challenger wins.
    \end{itemize}
    The game is called the IND-CPA game.
    
    Define the advantage of the attacker in the IND-CPA game as $\Adv^{\INDCPA}_{\Pi,\A}\coloneqq\left|\Pr\left[\hat{b}=b\mid\Game^{\INDCPA}_{\Pi,\A}\right]-1/2\right|$.
    We say $\Pi$ is IND-CPA secure, or $\Pi$ satisfies the IND-CPA security if $\Adv^{\INDCPA}_{\Pi,\A}$ is negligible, that is, there exists a negligible function $\epsilon$ such that $\Adv^{\INDCPA}_{\Pi,\A}\le\epsilon$.
\end{definition}

\begin{remark}
    Public-key encryption must at least satisfy the IND-CPA security to be used for secure communication.
    If the advantage of the attacker is not negligible, he/she can guess the encrypted message with a significant probability.
\end{remark}

\section{Provable Security for Encrypted Control}\label{sec:provable_security}

This section proposes a novel security notion of encrypted control based on the traditional provable-security notion in cryptography and reveals the relationship between the notions.

\subsection{Indistinguishability against parameter estimation attack}

This study considers an attack scenario in which an attacker attempts to estimate the controller parameters from the ciphertexts of the controller inputs and outputs.
The indistinguishability for encrypted control against the attack can be defined as follows.

\begin{definition}[IND-PEA]\label{def:ind-pea}
    Let $\Pi=(\KeyGen,\Enc,\Dec)$ be homomorphic encryption, and let $f_{\Pi}$ be an encrypted control law of \eqref{eq:control} with $\Pi$.
    Define a game between an attacker $\B=(\B_{1},\B_{2})$ and a challenger as follows.
    \begin{center}
        \fbox{
            \begin{minipage}{.9\columnwidth}
                $\Game^{\INDPEA}_{\Pi,f,\B}$:
                \begin{enumerate}
                    \renewcommand{\labelenumi}{\arabic{enumi}.}
                    \item $(\pk,\sk)\gets\KeyGen(\lambda)$
                    \item $(\bm{K}_{0},\bm{K}_{1},\sigma)\gets\B_{1}(\pk)$
                    \item $b\sample\{0,1\}$
                    \item $\bm{c_{K}}\gets\Enc_{\pk}(\bm{K}_{b})$
                    \item $\hat{b}\gets\B_{2}^{\mathcal{O}}(\sigma)$
                \end{enumerate}
            \end{minipage}
        }
    \end{center}
    \begin{itemize}
        \item\textbf{Setup}:
        The challenger computes keys $\pk$ and $\sk$ from a security parameter $\lambda$ using the key-generation function $\KeyGen$ and gives $\pk$ to the attacker (Line~$1$).
        \item\textbf{Challenge}:
        The attacker computes two different parameters $\bm{K}_{0},\bm{K}_{1}\in\M^{r}$, which are of identical size, and an intermediate state $\sigma$ using the polynomial-time algorithm $\B_{1}$ (Line~$2$).
        The attacker gives $\bm{K}_{0}$ and $\bm{K}_{1}$ to the challenger.
        The challenger randomly selects a bit $b\in\{0,1\}$ and computes a ciphertext $\bm{c_{K}}$ of $\bm{K}_{b}$ using the encryption function $\Enc$ with $\pk$ (Lines~$3$--$4$).
        The challenger sets $\bm{c_{K}}$ to the encrypted controller.
        \item\textbf{Guess}:
        The attacker guesses the selected parameter using the polynomial-time algorithm $\B_{2}$ and by querying an encrypted control oracle $\mathcal{O}$ polynomial times of $\lambda$ and outputs a bit $\hat{b}\in\{0,1\}$ (Line~$5$).
        The oracle receives a query $\bm{c_{y}}\in\C^{\ell}$ and returns an output $\bm{c_{u}}=f_{\Pi}(\bm{c_{y}};\bm{c_{K}})$.
        If $\hat{b}=b$, the attacker wins the game; otherwise, the challenger wins.
    \end{itemize}
    The game is called the \textit{indistinguishability under parameter estimation attack (IND-PEA)} game.
    
    Define the advantage of the attacker in the IND-PEA game as $\Adv^{\INDPEA}_{\Pi,f,\B}\coloneqq\left|\Pr\left[\hat{b}=b\mid\Game^{\INDPEA}_{\Pi,f,\B}\right]-1/2\right|$.
    We say $f_{\Pi}$ is IND-PEA secure, or $f_{\Pi}$ satisfies the IND-PEA security if $\Adv^{\INDPEA}_{\Pi,f,\B}$ is negligible.
\end{definition}

\figref{fig:ind_game} presents the schematic diagrams of the IND-CPA and IND-PEA games.
In the IND-CPA game of \figref{fig:ind-cpa}, an attacker guesses whether the plaintext, $m_{0}$ or $m_{1}$, has been encrypted from a given ciphertext $c$ of $m_{b}$.
Note that an attacker in this game can obtain the ciphertext of the target plaintext.
Conversely, an attacker in the IND-PEA game of \figref{fig:ind-pea} cannot obtain a ciphertext of the target parameter, $\bm{K}_{b}$.
An attacker can use the encrypted control oracle $\mathcal{O}$ with the encrypted parameter $\bm{c_{K}}=\Enc_{\pk}(\bm{K}_{b})$ to collect the encrypted controller outputs $\bm{c_{y}}$ for any input $\bm{c_{u}}$.
Therefore, the IND-PEA game simulates different situations within the IND-CPA game.

It should be noted that the oracle to be used by an attacker indicates that encrypted control is provided as a cloud service.
Most studies conducted on encrypted control consider cloud-based control systems~\cite{Nils20,Darup18_2,Alexandru18,Alexandru21_1,Teranishi21_4,Suh21_2}, and in such cases, an attacker can freely use an encrypted control algorithm.
Therefore, the use of the oracle in the IND-PEA game is reasonable when considering the security for the encrypted control.

\begin{figure}[t]
    \centering
    \subfigure[IND-CPA game.]{\includegraphics[scale=1]{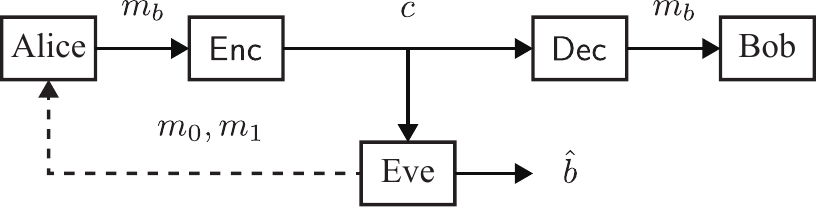}\label{fig:ind-cpa}}
    \subfigure[IND-PEA game.]{\includegraphics[scale=1]{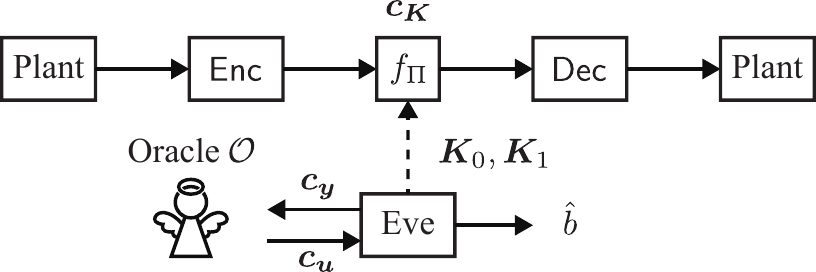}\label{fig:ind-pea}}
    \caption{Schematic pictures of IND-CPA and IND-PEA games.}
    \label{fig:ind_game}
    \vspace{-5mm}
\end{figure}

\subsection{Relationship between the security notions}

The security of encrypted control can be demonstrated using mathematical proofs based on the definition of IND-PEA security.
Furthermore, it gives the other benefit of discussing how much security strength is required to achieve secure encrypted control.
The following theorems disclose the relation between the IND-PEA security and IND-CPA security and clarify the security strength of the IND-PEA security.

\begin{theorem}\label{thm:weaker}
    Let $\Pi$ be homomorphic encryption and let $f_{\Pi}$ be an encrypted control law of \eqref{eq:control} with $\Pi$.
    If $\Pi$ is IND-CPA secure, $f_{\Pi}$ is IND-PEA secure.
\end{theorem}

\begin{proof}
    It is proved that if $f_{\Pi}$ is not IND-PEA secure, $\Pi$ is not IND-CPA secure.
    Define polynomial-time algorithms $\A_{1}$ and $\A_{2}$ as follows.
    \begin{center}
        \fbox{
            \begin{minipage}{.9\columnwidth}
                $\A_{1}(\pk)$:
                \begin{enumerate}
                    \renewcommand{\labelenumi}{\arabic{enumi}.}
                    \item $(m_{0},m_{1},\sigma)\gets\B_{1}(\pk)$
                    \item \textbf{return} $(m_{0},m_{1},\sigma)$
                \end{enumerate}
            \end{minipage}
        }
    \end{center}
    \begin{center}
        \fbox{
            \begin{minipage}{.9\columnwidth}
                $\A_{2}(c,\sigma)$:
                \begin{enumerate}
                    \renewcommand{\labelenumi}{\arabic{enumi}.}
                    \item \textbf{return} $\B_{2}^{\mathcal{O}}(\sigma)$
                \end{enumerate}
            \end{minipage}
        }
    \end{center}
    Substituting the algorithms into $\Game^{\INDCPA}_{\Pi,\A}$ in \defref{def:ind-cpa}, we obtain the following game.
    \begin{center}
        \fbox{
            \begin{minipage}{.9\columnwidth}
                $\Game_{\Pi,f,\B}$:
                \begin{enumerate}
                    \renewcommand{\labelenumi}{\arabic{enumi}.}
                    \item $(\pk,\sk)\gets\KeyGen(\lambda)$
                    \item $(m_{0},m_{1},\sigma)\gets\B_{1}(\pk)$
                    \item $b\sample\{0,1\}$
                    \item $c\gets\Enc_{\pk}(m_{b})$
                    \item $\hat{b}\gets\B_{2}^{\mathcal{O}}(\sigma)$
                \end{enumerate}
            \end{minipage}
        }
    \end{center}
    \begin{itemize}
        \item\textbf{Setup}:
        This phase is identical to the setup phase in \defref{def:ind-cpa}.
        \item\textbf{Challenge}:
        The attacker computes two different parameters $m_{0},m_{1}\in\M$, which are of identical size, and an intermediate state $\sigma$ using the polynomial-time algorithm $\B_{1}$ (Line~$2$).
        The attacker gives $m_{0}$ and $m_{1}$ to the challenger.
        The challenger randomly selects a bit $b\in\{0,1\}$ and computes a ciphertext $c$ of $m_{b}$ using the encryption function $\Enc$ with $\pk$ (Lines~$3$--$4$).
        The challenger sets $c$ to the encrypted controller.
        \item\textbf{Guess}:
        This phase is identical to the guess phase in \defref{def:ind-pea}.
    \end{itemize}
    Define the advantage of the attacker in the game $\Game_{\Pi,f,\B}$ as $\Adv_{\Pi,f,\B}\coloneqq\left|\Pr\left[\hat{b}=b\mid\Game_{\Pi,f,\B}\right]-1/2\right|$.
    Here, we consider the following lemmas to prove the theorem.
    
    \begin{lemma}\label{lem:1}
        $\left|\Adv^{\INDCPA}_{\Pi,\A}-\Adv_{\Pi,f,\B}\right|$ is negligible.
    \end{lemma}
    
    \begin{proof}
        The substitution to obtain the game $\Game_{\Pi,f,\B}$ does not change the probability of the game between the challenger and the attacker.
    \end{proof}

    \begin{lemma}\label{lem:2}
        $\left|\Adv_{\Pi,f,\B}-\Adv^{\INDPEA}_{\Pi,f,\B}\right|$ is negligible.
    \end{lemma}
    
    \begin{proof}
        The game $\Game_{\Pi,f,\B}$ corresponds to the IND-PEA game $\Game^{\INDPEA}_{\Pi,f,\B}$ with $r=1$, which is the dimension of the controller parameter.
        The advantage of the attacker does not change between these games.
    \end{proof}
    
    From \lemref{lem:1} and \lemref{lem:2}, it follows that the difference of advantages, $\left|\Adv^{\INDCPA}_{\Pi,\A}-\Adv^{\INDPEA}_{\Pi,f,\B}\right|$, is negligible.
    Furthermore, $\Adv^{\INDPEA}_{\Pi,f,\B}$ is not negligible because $f_{\Pi}$ is not IND-PEA secure, based on the assumption of this proof.
    Therefore, $\Adv^{\INDCPA}_{\Pi,\A}$ is not negligible, and so $\Pi$ is not IND-CPA secure.
\end{proof}

The theorem demonstrates that the IND-PEA security is weaker than or equal to the IND-CPA security.
This implies that secure encrypted control against the parameter estimation attack is achievable by using secure homomorphic encryption.
The following theorem presents one of the conditions so that the IND-PEA security is equivalent to the IND-CPA security.

\begin{theorem}\label{thm:equality}
    Let $\Pi$ be homomorphic encryption and let $f_{\Pi}$ be an encrypted control law of \eqref{eq:control} with $\Pi$.
    Suppose $f$ in \eqref{eq:control} is bijective for a fixed $\bm{y}\in\M^{\ell}$. 
    If $f_{\Pi}$ is IND-PEA secure and if the attacker can compute an encrypted control law of $f^{-1}$ with $\Pi$, then $\Pi$ is IND-CPA secure.
\end{theorem}

\begin{proof}
    Under the assumptions, we prove that if $\Pi$ is not IND-CPA secure, $f_{\Pi}$ is not IND-PEA secure.
    Let $f^{-1}_{\Pi}$ be an encrypted control law of $f^{-1}$ with $\Pi$.
    Since $f$ in \eqref{eq:control} is bijective for some fixed $\bm{y}\in\M^{\ell}$, it follows that

    \begin{equation}
        f^{-1}(\bm{y};f(\bm{y};\bm{K}))=\bm{K}
        \label{eq:finv}
    \end{equation}
    for all $\bm{K}\in\M^{r}$.
    Moreover, 
    \begin{equation}
        f^{-1}_{\Pi}(\bm{c_{y}};f_{\Pi}(\bm{c_{y}};\bm{c_{K}}))=\bm{c_{K}}
        \label{eq:ecinv}
    \end{equation}
    holds for $\bm{c_{y}}=\Enc_{\pk}(\bm{y})$ and $\bm{c_{K}}=\Enc_{\pk}(\bm{K})$.
    
    Define polynomial-time algorithms $\B_{1}$ and $\B_{2}$ as follows.
    \begin{center}
        \fbox{
            \begin{minipage}{.9\columnwidth}
                $\B_{1}(\pk)$:
                \begin{enumerate}
                    \renewcommand{\labelenumi}{\arabic{enumi}.}
                    \item $(K_{0},K_{1},\sigma)\gets\A_{1}(\pk)$
                    \item $\bm{K}_{0}\gets[K_{0}\ \cdots\ K_{0}]^{\top}$, $\bm{K}_{1}\gets[K_{1}\ \cdots\ K_{1}]^{\top}$
                    \item Fix $\bm{y}\in\M^{\ell}$ to satisfy \eqref{eq:finv}
                    \item $\sigma\gets(\sigma,\bm{y})$
                    \item \textbf{return} $(\bm{K}_{0},\bm{K}_{1},\sigma)$
                \end{enumerate}
            \end{minipage}
        }
    \end{center}
    \begin{center}
        \fbox{
            \begin{minipage}{.9\columnwidth}
                $\B_{2}^{\mathcal{O}}(\sigma)$:
                \begin{enumerate}
                    \renewcommand{\labelenumi}{\arabic{enumi}.}
                    \item $(\sigma,\bm{y})\gets\sigma$
                    \item $\bm{c_{y}}\gets\Enc_{\pk}(\bm{y})$
                    \item $\bm{c_{u}}\gets\mathcal{O}(\bm{c_{y}})$
                    \item $\hat{\bm{c}}_{\bm{K}}\gets f^{-1}_{\Pi}(\bm{c_{y}};\bm{c_{u}})$, $c_{K}\gets\hat{\bm{c}}_{\bm{K},1}$
                    \item \textbf{return} $\A_{2}(c_{K},\sigma)$
                \end{enumerate}
            \end{minipage}
        }
    \end{center}
    Substituting the algorithms into $\Game^{\INDPEA}_{\Pi,f,\B}$ in \defref{def:ind-pea}, we obtain the following game.
    \begin{center}
        \fbox{
            \begin{minipage}{.9\columnwidth}
                $\Game^{1}_{\Pi,f,\A}$:
                \begin{enumerate}
                    \renewcommand{\labelenumi}{\arabic{enumi}.}
                    \item $(\pk,\sk)\gets\KeyGen(\lambda)$
                    \item $(K_{0},K_{1},\sigma)\gets\A_{1}(\pk)$
                    \item $\bm{K}_{0}\gets[K_{0}\ \cdots\ K_{0}]^{\top}$, $\bm{K}_{1}\gets[K_{1}\ \cdots\ K_{1}]^{\top}$
                    \item Fix $\bm{y}\in\M^{\ell}$ to satisfy \eqref{eq:finv}
                    \item $\sigma\gets(\sigma,\bm{y})$
                    \item $b\sample\{0,1\}$
                    \item $\bm{c_{K}}\gets\Enc_{\pk}(\bm{K}_{b})$
                    \item $(\sigma,\bm{y})\gets\sigma$
                    \item $\bm{c_{y}}\gets\Enc_{\pk}(\bm{y})$
                    \item $\bm{c_{u}}\gets\mathcal{O}(\bm{c_{y}})$
                    \item $\hat{\bm{c}}_{\bm{K}}\gets f^{-1}_{\Pi}(\bm{c_{y}};\bm{c_{u}})$, $c_{K}\gets\hat{\bm{c}}_{\bm{K},1}$
                    \item $\hat{b}\gets\A_{2}(c_{K},\sigma)$
                \end{enumerate}
            \end{minipage}
        }
    \end{center}
    \begin{itemize}
        \item\textbf{Setup}:
        This phase is identical to the setup phase in \defref{def:ind-pea}.
        \item\textbf{Challenge}:
        The attacker computes two different plaintexts $K_{0},K_{1}\in\M$, which are of identical size, and an intermediate state $\sigma$ using the polynomial-time algorithm $\A_{1}$ (Line~$2$).
        Subsequently, the attacker constructs vectors $\bm{K}_{0},\bm{K}_{1}\in\M^{r}$ from $K_{0},K_{1}$ and gives them to the challenger (Line~$3$).
        Additionally, the attacker fixes $\bm{y}\in\M^{\ell}$ to satisfy \eqref{eq:finv} and updates $\sigma$ to $(\sigma,\bm{y})$ (Lines~$4$--$5$).
        The challenger randomly selects a bit $b\in\{0,1\}$ and computes a ciphertext $\bm{c_{K}}$ of $\bm{K}_{b}$ using the encryption function $\Enc$ with $\pk$ (Lines~$6$--$7$).
        The challenger sets $\bm{c_{K}}$ to the encrypted controller.
        \item\textbf{Guess}:
        The attacker takes out the intermediate state and $\bm{y}$ from $\sigma$ and computes the ciphertext $\bm{c_{y}}$ of $\bm{y}$ (Lines~$8$--$9$).
        Subsequently, the attacker queries the encrypted control oracle $\mathcal{O}$ to obtain a controller output $\bm{c_{u}}$ and computes a ciphertext $\hat{\bm{c}}_{K}$ by using the encrypted controller of $f^{-1}$ with $\Pi$ (Line~$10$--$11$).
        The attacker guesses the parameter that has been encrypted from the first element $c_{K}$ of the ciphertext vector $\hat{\bm{c}}_{K}$ using the polynomial-time algorithm $\A_{2}$ and outputs a bit $\hat{b}\in\{0,1\}$ (Line~$11$--$12$).
        If $\hat{b}=b$, the attacker wins the game; otherwise, the challenger wins.
    \end{itemize}
    Define the advantage of the attacker in the game $\Game^{1}_{\Pi,f,\A}$ as $\Adv^{1}_{\Pi,f,\A}\coloneqq\left|\Pr\left[\hat{b}=b\mid\Game^{1}_{\Pi,f,\A}\right]-1/2\right|$.
    Furthermore, we modify lines $3$--$5$ and $7$--$11$ of $\Game^{1}_{\Pi,f,\A}$ to obtain the following game.
    \begin{center}
        \fbox{
            \begin{minipage}{.9\columnwidth}
                $\Game^{2}_{\Pi,\A}$:
                \begin{enumerate}
                    \renewcommand{\labelenumi}{\arabic{enumi}.}
                    \item $(\pk,\sk)\gets\KeyGen(\lambda)$
                    \item $(K_{0},K_{1},\sigma)\gets\A_{1}(\pk)$
                    \item $b\sample\{0,1\}$
                    \item $c_{K}\gets\Enc_{\pk}(K_{b})$
                    \item $\hat{b}\gets\A_{2}(c_{K},\sigma)$
                \end{enumerate}
            \end{minipage}
        }
    \end{center}
    \begin{itemize}
        \item\textbf{Setup}:
        This phase is identical to the setup phase in \defref{def:ind-pea}.
        \item\textbf{Challenge}:
        The attacker computes two different plaintexts $K_{0},K_{1}\in\M$, which are of identical size, and an intermediate state $\sigma$ using the polynomial-time algorithm $\A_{1}$ (Line~$2$).
        The attacker gives $K_{0}$ and $K_{1}$ to the challenger.
        The challenger randomly selects a bit $b\in\{0,1\}$ and computes a ciphertext $c_{K}$ of $K_{b}$ using the encryption function $\Enc$ with $\pk$ (Lines~$3$--$4$).
        The challenger gives $c_{K}$ to the attacker.
        \item\textbf{Guess}:
        This phase is identical to the guess phase in \defref{def:ind-cpa}.
    \end{itemize}
    Define the advantage of the attacker in the game $\Game^{2}_{\Pi,\A}$ as $\Adv^{2}_{\Pi,\A}\coloneqq\left|\Pr\left[\hat{b}=b\mid\Game^{2}_{\Pi,\A}\right]-1/2\right|$.
    Here, we consider the following lemmas to prove the theorem.
    
    \begin{lemma}\label{lem:3}
        $\left|\Adv^{\INDPEA}_{\Pi,f,\B}-\Adv^{1}_{\Pi,f,A}\right|$ is negligible.
    \end{lemma}

    \begin{proof}
        The substitution for obtaining $\Game^{1}_{\Pi,f,\A}$ does not change the probability of the game between the challenger and the attacker.
    \end{proof}

    \begin{lemma}\label{lem:4}
        $\left|\Adv^{1}_{\Pi,f,A}-\Adv^{2}_{\Pi,A}\right|$ is negligible.
    \end{lemma}
    
    \begin{proof}
        It follows from \eqref{eq:ecinv} that $\hat{\bm{c}}_{\bm{K}}=\bm{c_{K}}=\Enc_{\pk}(\bm{K}_{b})=[\Enc_{\pk}(K_{b})\ \cdots\ \Enc_{\pk}(K_{b})]^{\top}$.
        Therefore, we can rewrite line $11$ in $\Game^{1}_{\Pi,f,\A}$ as line $4$ in $\Game^{2}_{\Pi,\A}$ without changing the probability of the games.
        The lines $3$--$5$ and $7$--$10$ of $\Game^{1}_{\Pi,f,\A}$ do not affect the advantage of the attacker in the game, and so we can remove the lines.
    \end{proof}
    
    \begin{lemma}\label{lem:5}
        $\left|\Adv^{2}_{\Pi,A}-\Adv^{\INDCPA}_{\Pi,\A}\right|$ is negligible.
    \end{lemma}
    
    \begin{proof}
        By definition, the game $\Game^{2}_{\Pi,A}$ is the same as the game $\Game^{\INDCPA}_{\Pi,A}$.
    \end{proof}
    
    It follows from \lemref{lem:3} to \lemref{lem:5} that the difference of advantages, $\left|\Adv^{\INDPEA}_{\Pi,f,\B}-\Adv^{\INDCPA}_{\Pi,\A}\right|$, is negligible.
    Furthermore, $\Adv^{\INDCPA}_{\Pi,\A}$ is not negligible since $\Pi$ is not IND-CPA secure based on the assumption of this proof.
    Therefore, $\Adv^{\INDPEA}_{\Pi,f,\B}$ is not negligible, and so $f_{\Pi}$ is not IND-PEA secure.
\end{proof}

Note that few controllers can satisfy the assumption in \thmref{thm:equality} in practice.
For example, in a static output feedback controller $\bm{u}=\bm{F}\bm{y}$, where $\bm{F}\in\M^{q\times\ell}$ is a feedback gain, the existence of inverse mapping $f^{-1}$ in \thmref{thm:equality} is equivalent to the condition that $q=r$, namely $\ell=1$.
Such a controller is a particular case of this control, and hence, it seems that the IND-PEA security is weaker than the IND-CPA security for most of the controllers formulated by \eqref{eq:control}.
Although one may think that this is a negative result, it suggests that homomorphic encryption satisfying IND-CPA security may not necessarily be required to prevent the parameter estimation attack.
In other words, the IND-PEA security might be achieved using more lightweight homomorphic encryption schemes than the conventional schemes to reduce the computation costs.
This implication is significant because encrypted control generally increases computational costs, and control systems require real-time computation.

\section{Further Development}\label{sec:further_development}

The security definition and analysis presented in this study were based on the assumption that the control law \eqref{eq:control} was employed.
Although the control law represents all static controllers with a controller parameter $\bm{K}$, there is scope for further expansion.
The results of this study can be expanded upon using a dynamical controller $\bm{x}_{t+1}=f(\bm{x}_{t},\bm{y}_{t};\bm{K})$, $\bm{u}_{t}=g(\bm{x}_{t},\bm{y}_{t};\bm{K})$, a time-varying controller $\bm{u}_{t}=f(\bm{y}_{t};\bm{K}_{t})$ instead of \eqref{eq:control}, where $\bm{x}$ is a controller state, and $t$ is a time step.
According to this change, the encrypted control oracle must be modified.

For example, the modified oracle for the dynamical controller must include a state updated by each query of an attacker and return an output based on the current state and input.
Note that the oracle state is hidden against the attacker.
The security of dynamical controllers under a parameter estimation attack requires to be analyzed while considering the state trajectory of the controller because the controller output depends on both the controller state and the parameter.
Such analysis is challenging because it needs to consider the effect of fundamental properties of dynamical systems, such as stability, controllability, and observability, on the provable security for encrypted control.
Moreover, an attacker may aim to disclose the (initial) state of an encrypted dynamical controller.
In this case, the security is defined through a game that differs from the IND-PEA game.

Furthermore, the security of controllers formulated by stochastic processes such as a Markov decision process can also be considered.
An example of this type of controller is reinforcement learning, which is an application of encrypted control~\cite{Suh21_2}.
In this case, the state of the environment randomly moves in each time step, and an action, which is an input for the environment, is determined by a policy that maximizes the expected cumulative reward.
An attacker for an agent of reinforcement learning may be interested in the policy rather than the controller parameter.
Further research on provably secure encrypted control involves considering such stochastic controllers.

\section{Conclusions}\label{sec:conclusions}

This study defined a novel security for encrypted control under a parameter estimation attack in terms of provable security in cryptography and analyzed its security strength.
The definition enables us to prove the security for encrypted control by mathematical procedures.
The analysis revealed that encrypted control is secure if homomorphic encryption used for the control is secure.
This result means that most existing encrypted controls are secure under a parameter estimation attack.
The proposed security can be extended in the future to encrypted control using dynamical controllers.

\bibliographystyle{IEEEtran}
\bibliography{encrypted_control_and_optimization,others}

\end{document}